%% file: Paper.tex
\documentclass[conference,twocolumn]{IEEEtran}
\IEEEoverridecommandlockouts

\usepackage[utf8]{inputenc} 
\usepackage[T1]{fontenc}
\usepackage{url}
\usepackage{ifthen}
\usepackage[cmex10]{amsmath} % Use the [cmex10] option to ensure complicance
\newboolean{JOURNAL} % set the Boolean variable after the package "ifthen"
\setboolean{JOURNAL}{true} % true for the journal paper
\newboolean{TWOCOLUMN} % set the Boolean variable after the package "ifthen"
\setboolean{TWOCOLUMN}{true} % true for the two-column version
\newboolean{NOTE} % set the Boolean variable for the helpful notes
\setboolean{NOTE}{true} % true for the version with notes
\newboolean{arXiv} % set the Boolean variable for the arXiv/extended version
\setboolean{arXiv}{true} % true for the arXiv version
\usepackage{mleftright}       % fix to wrong spacing of \left-,
\mleftright                   % \middle- \right-commands 
\usepackage{balance}
\usepackage{graphicx}         % provides \includegraphics{...} to
\usepackage{booktabs}         % fixes poor spacing in tables and
\usepackage{array}
\usepackage{comment}
\usepackage[url,hyperrefblack,notheorems,IEEEtran]{research17} % customized LaTeX commands
\usepackage{amsmath,amssymb,amsfonts,amsbsy}
\usepackage[colorinlistoftodos, textsize=footnotesize]{todonotes} % add inline or margin comments

\usepackage{amsthm} % Note that amsthm must be loaded manually AFTER option *notheorems! 
\newtheorem{theorem}{\mytheoremname}
\newtheorem{lemma}{\mylemmaname}
\newtheorem{corollary}{\mycorollaryname}

\newtheorem{proposition}{\mypropositionname}

\newtheorem{definition}{\mydefinitionname}

\usepackage[capitalize]{cleveref}
\usepackage{tikz}
\usetikzlibrary{spy} % <-- for zoom in/out on a part of a curve:
\usepackage{pgfplots}
\allowdisplaybreaks

\newcommand{\Hwt}[1]{\wH\left(#1\right)} % the (Hamming) weight function
\allowdisplaybreaks

\usepackage{array}
\usepackage{multirow}
\usepackage{enumitem} % required for \begin{description}
\usepackage{nicematrix} % for adding vertical lines in a matrix
\usepackage{soul} % hyphenation for letterspacing, underlining, and more
\usepackage[colorinlistoftodos,textsize=footnotesize]{todonotes} % add inline or margin comments
\definecolor{darkgreen}{rgb}{0, 0.60, 0}
\usepackage{bbm}
\usepackage{import}
\usepackage{standalone}
\usepackage{packages/commands}
\usepackage{packages/colors}

\usepackage[caption=false,font=footnotesize]{subfig}
\begin{document}
\title{A Spectrum-Based Converse for Quantum\\ State Discrimination and Its Applications to Classical-Quantum Channel Coding\thanks{This work was supported in part by the Research Council of Norway (RCN) under the NISQEC project (grant no.~$357698$).}}

% %%% Single author, or several authors with same affiliation:

% \ifthenelse{\boolean{JOURNAL}}{
% \author{Tamás~Havas,~\IEEEmembership{Student Member,~IEEE},
% Hsuan-Yin~Lin,~\IEEEmembership{Senior Member,~IEEE},
% Eirik~Rosnes,~\IEEEmembership{Senior Member,~IEEE}
% \thanks{Tamás~Havas, Hsuan-Yin~Lin, and Eirik~Rosnes are with Simula UiB, N--5006 Bergen, Norway. Email: \{tamas, lin, eirikrosnes\}@simula.no.}
% }
% }{
\author{\IEEEauthorblockN{Tam{\'a}s Havas \IEEEauthorrefmark{2}\IEEEauthorrefmark{3}, Hsuan-Yin Lin\IEEEauthorrefmark{2}, and Eirik Rosnes\IEEEauthorrefmark{2}}
\IEEEauthorblockA{\IEEEauthorrefmark{2}Simula UiB, N--5006 Bergen, Norway\\
\IEEEauthorrefmark{3}Department of Informatics, University of Bergen, N--5020 Bergen, Norway\\
Emails: \{tamas, lin, eirikrosnes\}@simula.no}}
% }

%%%%%   
%%%%%   Several authors with up to three affiliations:
%   \author{%
%   \IEEEauthorblockN{Stefan M.~Moser}
%   \IEEEauthorblockA{ETH Zürich\\
%   ISI (D-ITET), ETH Zentrum\\
%   8092 Zürich, Switzerland\\
%   moser@isi.ee.ethz.ch}
%   \and
%   \IEEEauthorblockN{Albus Dumbledore and Harry Potter}
%   \IEEEauthorblockA{Hogwarts School of Witchcraft and Wizardry\\
%   Hogwarts Castle\\ 
%   1714 Hogsmeade, Scotland\\
%   \{dumbledore, potter\}@hogwarts.edu}
% }
%   
%%%%%   
%%%%%   Many authors with many affiliations:
%   \author{%
%   \IEEEauthorblockN{Albus Dumbledore\IEEEauthorrefmark{1},
%   Olympe Maxime\IEEEauthorrefmark{2},
%   Stefan M.~Moser\IEEEauthorrefmark{3}\IEEEauthorrefmark{4},
%   and Harry Potter\IEEEauthorrefmark{1}}
%   \IEEEauthorblockA{\IEEEauthorrefmark{1}%
%   Hogwarts School of Witchcraft and Wizardry,
%   1714 Hogsmeade, Scotland,
%   \{dumbledore, potter\}@hogwarts.edu}
%   \IEEEauthorblockA{\IEEEauthorrefmark{2}%
%   Beauxbatons Academy of Magic,
%   1290 Pyrénées, France,
%   maxime@beauxbatons.fr}
%   \IEEEauthorblockA{\IEEEauthorrefmark{3}%
%   ETH Zürich, ISI (D-ITET), ETH Zentrum, 
%   CH-8092 Zürich, Switzerland,
%   moser@isi.ee.ethz.ch}
%   \IEEEauthorblockA{\IEEEauthorrefmark{4}%
%   National Yang Ming Chiao Tung University (NYCU), 
%   Hsinchu, Taiwan,
%   moser@isi.ee.ethz.ch}
% }
%   

\maketitle

\begin{abstract}
  We investigate converse bounds on the average decoding error probability in finite-blocklength classical-quantum channel coding. We first present a lower bound for multiple quantum hypothesis testing in terms of pairwise trace distances and derive a corresponding fidelity bound. We then obtain a spectrum-based converse that depends only on the a priori probabilities and spectra of the states. We show that the converse bound remains tight for the quantum depolarizing channel under suitable conditions. For codes with product-state outputs, this converse takes an explicit form involving products of output-state eigenvalues. We apply it to binary codes over the quantum amplitude damping channel using the input states $\vert+\rangle$ and $\vert-\rangle$. For this setting, we also discuss a normal approximation to the spectrum-based converse in the large blocklength regime. In all numerical examples considered, the spectrum-based converse is tighter than the other converse bounds at low noise levels.
\end{abstract}% xxx

\section{Introduction} 
\label{sec:introduction}  
% \input{sections/Introduction}
% \todo[inline]{Lin}

In quantum Shannon theory, a fundamental problem is to characterize the reliable transmission of classical messages over a noisy quantum channel. The Holevo--Schumacher--Westmoreland coding theorem characterizes the classical capacity, namely, the maximum rate achievable with vanishing error probability as the number of channel uses tends to infinity~\cite{Holevo98_1, SchumacherWestmoreland97_1}. For a given code length and a fixed number of messages, however, capacity alone does not specify the decoding error probability. This finite-blocklength question is particularly relevant when the number of channel uses is limited~\cite{PolyanskiyPoorVerdu10_1, MatthewsWehner14_1}.

In classical--quantum (CQ) channel coding, a classical code and a set of quantum input states jointly determine the states at the output of a noisy quantum channel. The receiver then performs a measurement to decode the classical message. Our previous work on the quantum amplitude damping channel (ADC) showed that both the code and the measurement affect decoding performance at finite blocklengths~\cite{HavasLinRosnesLai25_1}. In particular, collective measurements can improve performance for coded transmission, whereas uncoded transmission does not benefit from them in that setting. To characterize the fundamental performance limit of a given coding scheme, we therefore seek a lower bound on its average decoding error probability that applies to any measurement at the receiver.

For a fixed set of channel output states, finding the minimum average error probability is a multiple quantum hypothesis testing problem. The optimal measurement can be computed numerically for small systems, but the output dimension grows with the code length. Finite-blocklength converses based on binary hypothesis testing provide general benchmarks for quantum channel coding~\cite{MatthewsWehner14_1, WangFangTomamichel19_1}. Other bounds for multiple-state discrimination use pairwise trace distances or relations between the states' eigenspaces~\cite{Qiu08_1, Li16_1}. For symmetric CQ channels, the connection with generalized perfect codes gives an exact characterization in particular cases~\cite{BlascoCollVazquez-VilarFonollosa22_1}. These approaches generally require the full output states or pairwise relations between them. A converse based only on the spectrum of each output state, i.e., the set of its eigenvalues, would require less information about the states. This raises the question of what lower bound on the average decoding error probability can be obtained using only the a priori probabilities and the output states' spectra.

% \todo[inline, author=Lin]{WIP. Will finish writing the contributions section after I have carefully reviewed the results.}
This work studies converse bounds on the average decoding error probability for CQ channel coding through the corresponding quantum state discrimination problem. Our main contributions are listed as follows.
\begin{enumerate}[label=\roman*),leftmargin=*]
\item We first derive a lower bound on the error probability in terms of pairwise trace distances and a related fidelity bound. The trace-distance bound improves the corresponding bound of Qiu~\cite{Qiu08_1}.
  
\item We obtain a spectrum-based upper bound on the success probability that depends only on the a priori probabilities and the eigenvalues of the states. We give conditions for the bound to be tight and examine when this tightness is preserved under a depolarizing channel. For product-state outputs, the bound can be written in terms of products of single-use eigenvalues and coincides with a sphere-packing bound for an auxiliary classical channel.
  
\item We apply the spectral converse to binary codes over the quantum amplitude damping channel with input states $\lvert+\rangle$ and $\lvert-\rangle$. Their output spectra yield an explicit finite-blocklength bound. We also discuss a normal approximation to the bound as the code length grows. For the repetition, single-parity-check, and shortened Hamming codes considered numerically, the spectral converse gives a tighter bound than the other converses at low noise levels.
\end{enumerate}
% xxx

\section{Preliminaries and Setup} 
\label{sec:preliminaries}
\subsection{Notation} 
\label{sec:notation}

We denote the sets of natural numbers, integers, rationals, reals, and the complex numbers as $\Nbb$, $\Zbb$, $\Qbb$, $\Rbb$, and $\Cbb$, respectively. The imaginary unit is denoted by $\ieu$. For $i,j \in \Zbb$ with $i \leq j$, let $[i:j] \eqdef \{i, i+1, \ldots, j\}$. We use the Dirac delta notation $\delta_{ij} = 1$ if $i=j$ and $\delta_{ij} = 0$ for $i \neq j$.

\emph{Column} vectors are represented by boldfaced letters, e.g., $\abm$ or $\abm_i$, and their entries with normal sized letters $a_j$ or $a_{i,j}$,  respectively. The all-one vector is denoted by $\vect{1}$, and the context often implies the dimension. Matrices are represented by capital sans-serif letters and sets by calligraphic uppercase letters, respectively, e.g., $\Xsf$ and $\Xcal$. The identity matrix and the all-zero matrix are denoted by $\Isf$ and $\Osf$, respectively. 

An $(n,\Meu)_q$ or $(n,\Meu,d)_q$ classical code $\Cscr$ (linear or nonlinear) over a finite field $\Fbb_q$ is defined as a subset of $\Fbb_q^n$ containing $\const{M}$ codewords and having a minimum distance $d$. If the code $\Cscr$ is linear with dimension $k$, it is referred to as an $[n,k]_q$ or $[n,k,d]_q$ code $\Cscr$.
We denote the \emph{Hamming weight} of a vector $\xbm \in \Fbb_q^n$ by $\Hwt{\xbm}$ and use $\trans{(\cdot)}$ to represent the transpose of its argument.

The set of linear operators over an $\Delta$-dimensional Hilbert space $\Hcal_{\Delta}$ is denoted by $\Lscr(\Hcal_{\Delta})$, the subscript may be omitted if not significant. The set of density operators is written as $\Dscr(\Hcal) \subseteq \Lscr(\Hcal)$ and operators are denoted by lower case Greek letters. The Hermitian conjugate of an argument is represented by $\hermi{(\cdot)}$. For $i \in [0:\Delta-1]$, $\ket{i}$ represents the standard basis of an $\Delta$-dimensional Hilbert space. A series of tensor products is denoted by
\begin{IEEEeqnarray*}{c}
  \bigotimes_{j \in [1:n]} \Asf_j \eqdef \Asf_1 \otimes \ldots \otimes \Asf_n,
\end{IEEEeqnarray*}
and $\Asf^{\otimes n} = \bigotimes_{j \in [1:n]} \Asf$ is the $n$-fold tensor product of an operator with itself. The tensor product operations can be applied to operators, channels, and Hilbert spaces. The spectrum of an operator $\Asf$ is denoted by $\tn{spec}(\Asf)$, and it is a non-increasing vector unless specifically stated otherwise.

Any Hermitian operator $\Hsf$ can be written as $\Hsf = \Asf - \Bsf$ for some $\Asf, \Bsf$ orthogonal positive semi-definite operators in the same space, then we define the operator absolute value as $\absol{\Hsf} \eqdef \Asf + \Bsf$. The projector onto the support of $\Asf$ is denoted by $\brace{\Hsf}_+$ and onto the kernel of $\Hsf$ by $\brace{\Hsf}_0$. We adopt the common Dirac notation throughout this work, where $\ket{\psi}$ denotes a pure quantum state expressed as a column vector. For example, the spectral decomposition of a Hermitian operator $\Hsf$ can be written as $\Hsf = \sum_{a \in \tn{spec}(\Hsf)} a \cdot \ketbra{a}{a}$, where $\tn{spec}(\Hsf)$ is the spectrum and the eigenvalues label the unit eigenvectors. The squared fidelity between density operators is $F(\varphi, \rho)\eqdef\operatorname{tr}^2\bigl(\eabs{\sqrt{\varphi} \sqrt{\rho}}\bigr)$ and its square root is denoted by $\sqrt{F}(\Asf,\Bsf)$, and the binary entropy function is $\Hb$.

% === This helps Overleaf to understand how to collapse subsections === %

\subsection{Quantum Hypothesis Testing (QHT)}
\label{sec:QHT}
We look at CQ communication, where we represent classical information $[1:\Meu]$ with a set of density operators $\Scal = \brace{\varphi_1, \dots, \varphi_\Meu}$ subject to a noisy quantum channel $\Ncal$ defined by its Kraus operators $\Kcal_\Ncal$ as
\begin{equation*}
  \Ncal(\varphi) = \sum_{\Ksf \in \Kcal_\Ncal} \Ksf \varphi \hermi{\Ksf},
\end{equation*}
where $\sum_{\Ksf \in \Kcal_\Ncal} \hermi{\Ksf} \Ksf = \Isf$. We will denote the action of a noisy quantum channel on a set as a composition
\begin{equation*}
  \Ncal \circ \Scal = \brace{\Ncal(\varphi_1), \dots, \Ncal(\varphi_\Meu)}.
\end{equation*}
A measurement is described by a set of non-negative operators $\Dcal = \brace{\Lambda_1, \dots, \Lambda_\Meu}$ called a positive-operator valued measure (POVM), such that $\sum_{\Lambda \in \Dcal} \Lambda = \Isf$ sums to the identity. The success probability of discriminating between the elements of a set $\Scal$ using a measurement $\Dcal$ with a priori distribution $\pbm$ is given by the Born rule
\begin{equation*}
  P_\ctn(\pbm, \Scal, \Dcal) = \sum_i p_i \trace{\varphi_i \Lambda_i}.
\end{equation*}
With some abuse of notation, whenever we omit an argument, we imply an optimization over it, e.g.,
\begin{equation*}
  P_\ctn(\Scal) = \sup_{\pbm, \Dcal} P_\ctn(\pbm, \Scal, \Dcal).
\end{equation*}
We are interested in the questions: What is the optimal success probability/measurement under the uniform a priori $\nicefrac{\bone}{\Meu}$? I.e.,
\begin{IEEEeqnarray*}{c}
  P_\ctn\paren{\frac{\bone}{\Meu}, \Ncal \circ \Scal} , \quad \Dcal^* = \argmin_{\Dcal} P_\ctn\paren{\frac{\bone}{\Meu}, \Ncal \circ \Scal, \Dcal}.
\end{IEEEeqnarray*}
Alternatively, we will sometimes use the error probability $P_\etn(\cdot) = 1 - P_\ctn(\cdot)$.

The next theorem gives the exact error probability.
\begin{lemma}[{\cite[Th.~1]{BlascoCollVazquez-VilarFonollosa22_1}}] \label{lem:mqht_to_bqht}
  The error probability of a discrimination problem is given by
  \begin{equation}
    \label{eq:bqht_max}
    P_\etn(\pbm, \Scal) = \max_{\mu \in \Dcal(\Hcal)} \alpha_{\frac{1}{\Meu}} \paren{\varphi \middle|\middle| \frac{\Isf}{\Meu} \otimes \mu},
  \end{equation}
  where $\varphi = \diag{p_1 \varphi_1, \dots, p_\Meu \varphi_\Meu}$ and $\alpha$ is the optimal type I error for a binary quantum hypothesis testing problem
  \begin{equation}
    \label{eq:bqht_type_I}
    \alpha_\beta (\psi_0 || \psi_1) = \inf_{\substack{0 \leq \Lambda_0 \leq \Isf \\ \trace{\psi_1 \Lambda_0} \leq \beta}} 1 - \trace{\psi_0 \Lambda_0}.
  \end{equation}
\end{lemma}

Quasi-perfect codes achieve the optimal error probability for the parameters they have. \footnote{We use the weaker definition presented in the arXiv version~\cite{BlascoCollVazquez-VilarFonollosa20_1sub}, not in the published version~\cite{BlascoCollVazquez-VilarFonollosa22_1}. By this definition, we are describing an equiprobable channel coding setup, but we can generalize it to the operators $p_i \varphi_i$.
}
\begin{definition}[Quasi-Perfect Code (QPC)]
  \label{def:quasi-perfect_code}
  A set $\brace{\varphi_m}_{m \in \Meu} \subset \Dscr(\Hcal)$ forms a quasi-perfect code, if there exist $\mu \in \Dscr(\Hcal)$ and $t \in \Rbb$, such that the projectors $\brace{\varphi_m - t\mu}_+$ are pair-wise orthogonal and
  \begin{equation*}
    \sum_{m \in [1:\Meu]} \brace{\varphi_m - t\mu}_+ + \brace{\varphi_m - t\mu}_0 \geq \Isf.
  \end{equation*}
\end{definition}
Once the operator $t \mu$ is known, calculating the error probability and the corresponding measurement reduces to the eigenvalue problem of the operators $\varphi_m - t\mu$. However, finding $t \mu$ is as hard as solving the original optimization problem in general. As a result, we are not considering QPCs in this paper.

% === This helps Overleaf to understand how to collapse subsections === %

\subsection{Channel Coding over Classical-Quantum Channels}
\label{sec:classical_coding}

Classical channel coding for $n$ uses of a noisy quantum channel $\Ncal^{\otimes n}$ is a set of density operators $\brace{\varphi_m}_{m \in [1:\Meu]} \subset \Hcal^{\otimes n}$. We are interested in the case where the channel inputs are defined by a classical error-correcting code $\Cscr$ with parameters $(n, \Meu)_q$. First, let us fix a set of input symbols $\Scal = \brace{\rho_i}_{i \in [0:q-1]}$, then the channel input for message $m \in [1:\Meu]$ is given by
\begin{equation*}
  \varphi_m \eqdef \bigotimes_{l \in [1:n]} \rho_{x_{m,l}},
\end{equation*}
where $\xbm_m \in \Cscr$ is the code word corresponding to message $m$. In this case, the channel outputs are calculated as a tensor product
\begin{equation*}
  \Ncal^{\otimes n} (\varphi_m) = \bigotimes_{l \in [1:n]} \Ncal(\rho_{x_{m,l}}) = (\Ncal \circ \Scal)_{\xbm_m}.
\end{equation*}
For the set of density operators generated by a coding scheme, we use the notation $\Scal_\Cscr$, so the noisy channel outputs are denoted by
\begin{equation*}
  (\Ncal \circ \Scal)_\Cscr.
\end{equation*}

% === This helps Overleaf to understand how to collapse subsections === %

\subsection{QHT Bounds in the Literature}

Usually, the interesting use-cases have code lengths $n$ that are too large to consider for the optimization in \cref{lem:mqht_to_bqht} or the related semi-definite programming, because the dimension of the Hilbert space $N = q^n$ grows exponentially in $n$. As we will see the explicitly defined bounds in the literature use trace distance because it is closely related to measurements. However, this metric requires the calculation of the spectra of $N$-by-$N$ operators. To counter this, we will turn these bounds (that we will call trace-distance bounds) into bounds in terms of fidelity $F$, we will refer to these as fidelity bounds. The main idea behind this flip is the multiplicativity of $F$
\begin{equation} \label{eq:multiplicative_fidelity}
    F(\rho_0 \otimes \rho_1, \psi_0 \otimes \psi_1) = F(\rho_0, \psi_0) \cdot F(\rho_1, \psi_1).
\end{equation}
This paper focuses on and compares fidelity bounds and only mentions the trace bound that we converted to fidelity-based ones. To switch between the two distances, we will use the following lemma.
\begin{lemma}
  \label{lem:fidelity_vs_trace}
  For any two operators $\Lscr(\Hcal) \ni \Asf, \Bsf \geq 0$ with traces $a$ and $b$, respectively, if $a, b \leq 1$ then
  \begin{IEEEeqnarray*}{rCl}
    \frac{\max(a,b)}{2} - \sqrt{F}(\Asf, \Bsf) & \leq &\frac{\trace{\eabs{\Asf - \Bsf}}}{2}
    \\
    & \leq &\sqrt{\paren{\frac{a+b}{2}}^2 - F(\Asf, \Bsf)}.
  \end{IEEEeqnarray*}
\end{lemma}
\begin{IEEEproof}
  See (A.11) in the supplementary material of~\cite{ChengDatta-etal25_1}.
\end{IEEEproof}
    
The state-of-the-art when it comes to converse bounds on multiple QHT is from Qiu~\cite{Qiu08_1} and Li~\cite{Li16_1}. The following trace-distance bound on the error probability is from~\cite{Qiu08_1}.
\begin{theorem}[{\cite[Th.~1]{Qiu08_1}}]
  \label{thm:Qiu_trace_bound}
  Given a discrimination problem between the quantum states $\Scal = \brace{\varphi_1, \dots, \varphi_\Meu} \subset \Dscr(\Hcal)$ with a priori distribution $\pbm$, the error probability is bounded by
  \begin{IEEEeqnarray*}{rCl}
    P_\etn(\pbm, \Scal)& \geq &\frac{1}{2} - \frac{1}{2 (\Meu - 1)}\sum_{1 \leq i < j \leq \Meu}\trace{\eabs{p_i \varphi_i - p_j \varphi_j}}.
  \end{IEEEeqnarray*}
\end{theorem}

Let us apply \cref{lem:fidelity_vs_trace} to create a fidelity bound out of \cref{thm:Qiu_trace_bound} in order to utilize \cref{eq:multiplicative_fidelity}.
\begin{corollary}
  \label{cor:Qiu_fidelity_bound}
  Given a discrimination problem between the quantum states $\Scal = {\varphi_1, \dots, \varphi_\Meu} \subset \Dscr(\Hcal)$ with a priori distribution $\pbm$, the error probability is bounded by
  \begin{IEEEeqnarray}{rCl}
    \IEEEeqnarraymulticol{3}{l}{%
      P_\etn(\pbm, \Scal)\geq \frac{1}{2} - \frac{1}{2 (\Meu - 1)}}
    \nonumber\\*\qquad\qquad%
    & \cdot &\sum_{1 \leq i < j \leq \Meu}\sqrt{p_i^2 + 2 (1 - 2F(\varphi_i, \varphi_j)) p_i p_j + p_j^2}.\IEEEeqnarraynumspace
  \end{IEEEeqnarray}
\end{corollary}

In~\cite[Eq.~(5)]{Li16_1}, a lower bound on the average error probability for multiple quantum state discrimination in terms of pairwise spectral overlaps was derived. 
\begin{theorem}[{\cite[Eq.~(5)]{Li16_1}}]
  \label{thm:Li_fidelity_bound}
  Given a discrimination problem between the quantum states $\Scal = {\varphi_1, \dots, \varphi_\Meu} \subset \Dscr(\Hcal)$ with a priori distribution $\pbm$, the error probability is bounded by
  \begin{IEEEeqnarray}{rCl}
    P_\etn(\pbm, \Scal) &\geq& \frac{1}{2 (\Meu - 1)} \sum_{1 \leq i < j \leq \Meu} \sum_{\substack{a_i \in \tn{spec}(\varphi_i) \\ a_j \in \tn{spec}(\varphi_j)}} 
    \IEEElinebreak{} \min(p_i a_i, p_j a_j) F(\Asf_i, \Asf_j)
  \end{IEEEeqnarray}
  where $\Asf_\ell = \ketbra{a_\ell}{a_\ell}$ is the pure state corresponding to the eigenvalue $a_\ell$ of $\varphi_\ell$ for $\ell \in [1:\Meu]$.
\end{theorem}
% \begin{proof}
%   See Appendix~\ref{sapp:Lis_fidelity_bound}.
% \end{proof}

% === This helps Overleaf to understand how to collapse subsections === %

\subsection{Spectra for the Amplitude Damping Channel}

The ADC $\Ncal_\gamma$ models a quantum system with a ground- and an excited state, and it is defined by its Kraus operators
\begin{equation*}
  \Kcal_{\Ncal_\gamma} = \brace{\begin{bmatrix} 1 & 0 \\ 0 & \sqrt{1 - \gamma} \end{bmatrix} ,\quad \begin{bmatrix} 0 & \sqrt{\gamma} \\ 0 & 0 \end{bmatrix}}.
\end{equation*}

\begin{lemma}
  \label{lem:adc_specrtum}
  Given two orthogonal single-qubit states as channel inputs
  \begin{equation*}
    \rho_0 = \begin{bmatrix} a & b e^{\ieu f} \\ b e^{-\ieu f} & 1-a \end{bmatrix} ,\quad 
    \rho_1 = \begin{bmatrix} 1-a & -b e^{\ieu f} \\ -b e^{-\ieu f} & a \end{bmatrix},
  \end{equation*}
  where $b = \sqrt{a(1-a)}$, and the ADC $\Ncal_\gamma$, the spectra of the output is
  \begin{align*}
    \tn{spec}(\Ncal_\gamma(\rho_0)) &= \trans{\paren{\frac{1 + \sqrt{1 - 4c_a}}{2}, \frac{1 - \sqrt{1 - 4c_a}}{2}}} \\
    \tn{spec}(\Ncal_\gamma(\rho_1)) &= \trans{\paren{\frac{1 + \sqrt{1 - 4c_{1-a}}}{2}, \frac{1 - \sqrt{1 - 4c_{1-a}}}{2}}},
  \end{align*}
  where $c_a = \gamma (1 - \gamma) (1-a)^2$.
\end{lemma}

\section{Main Results}
\label{sec:spectral_converse}
Let $\Scal = \brace{\varphi_1, \dots, \varphi_\Meu} \subset \Dscr(\Hcal_N)$ be a set of symbols with the sorted spectra $\abm_m = \tn{spec}(\varphi_m)$, spectral decomposition $\varphi_m = \sum_{j \in [1:N]} a_{m, j} \ketbra{a_{m, j}}{a_{m, j}}$, and a priori distribution $\pbm$. Throughout \cref{sec:spectral_converse} we consider the discrimination problem between the states of $\Scal$ and use the variables defined here.

\subsection{A Converse Bound on Error Probability}

Along the lines of~\cref{thm:Qiu_trace_bound} one can find the following converse.
\begin{theorem}
  \label{thm:Gemini_trace_bound}
  The error probability of any measurement is bounded by
  \begin{equation*}
    P_\etn(\pbm, \Scal) \geq 1 - \frac{1}{\Meu} - \frac{1}{\Meu} \sum_{1 \leq i < j \leq \Meu}\trace{\eabs{p_i \varphi_i - p_j \varphi_j}}.
  \end{equation*}
\end{theorem}
\begin{IEEEproof}
  See Appendix~\ref{sapp:Gemini_trace_bound}.
\end{IEEEproof}

Using \cref{lem:fidelity_vs_trace} as in the case of \cref{thm:Qiu_trace_bound} we turn this into a fidelity bound.

\begin{corollary}
  \label{cor:Gemini_fidelity_bound}
  The error probability of any measurement is bounded by
  \begin{IEEEeqnarray}{rCl}
    P_\etn(\pbm, \Scal) &\geq& 1 - \frac{1}{\Meu} - \frac{1}{\Meu} \sum_{1 \leq i < j \leq \Meu}
    \IEEElinebreak{} \sqrt{p_i^2 + 2 (1 - 2F(\varphi_i, \varphi_j)) p_i p_j + p_j^2}.
  \end{IEEEeqnarray}
\end{corollary}

Note that \cref{thm:Gemini_trace_bound} is always dominates \cref{thm:Qiu_trace_bound}, and \cref{cor:Gemini_fidelity_bound} is always dominates \cref{cor:Qiu_fidelity_bound}.\footnote{In general, $1 - \frac{1}{\Meu} - \frac{x}{\Meu} \geq \frac{1}{2} - \frac{x}{2 (\Meu - 1)}$ if and only if $x \leq \Meu - 1$, and the maximal value that both the sum of trace distances and sum of fidelity-containing square roots is $\Meu-1$.}

% === This helps Overleaf to understand how to collapse subsections === %

\subsection{A Converse Bound on Success Probability}

\cref{lem:mqht_to_bqht} does not give an explicit expression for the error/success probability in general, and \cref{def:quasi-perfect_code} requires $t\cdot\mu$ for calculating the probability. The following bound is an instance of the optimization in \cref{lem:mqht_to_bqht}, where we only have to calculate the spectra, instead of optimizing the POVM or something equivalent.
\begin{theorem}[Spectral converse] \label{thm:spectral_converse}
  Let $\Ical_{\Scal}$ be an index set with size $|\Ical_{\Scal}| = N$ and the property
  \begin{equation} \label{eq:def_of_Ical}
    \forall (i, j) \in \Ical_{\Scal}, \forall (i', j') \in \Ical_{\Scal}^c: p_i \cdot a_{i, j} \geq p_{i'} \cdot a_{i', j'}.
  \end{equation}
  For any POVM $\Dcal = \brace{\Lambda_i}_{i \in [1:\Meu]}$, the success probability is upper-bounded by
  \begin{equation*}
    P_\ctn(\pbm, \Scal, \Dcal) \leq \sum_{(i,j) \in \Ical_{\Scal}} p_i a_{i, j} \leq 1,
  \end{equation*}
  where the first equality holds in the case 
  \begin{equation*} \label{eq:sc_optimality}
    \forall i \in [1:\Meu]: \Lambda_i = \sum_{j: (i,j) \in \Ical_{\Scal}} \ketbra{a_{i,j}}{a_{i,j}}.
  \end{equation*}
\end{theorem}
\begin{IEEEproof}
  See Appendix~\ref{sapp:spectral_converse}
\end{IEEEproof}

The spectral converse is always defined and can be tight, meaning that it is achievable. In these cases, we are dealing with a quasi-perfect code.
\begin{corollary} \label{thm:spectral_converse_tight}
  If the spectral converse bound is tight, i.e.,
  \begin{equation*}
    P_\ctn(\pbm, \Scal) = \sum_{(i,j) \in \Ical_{\Scal}} p_i a_{i, j},
  \end{equation*}
  then $\Scal$ forms a quasi-perfect code.
\end{corollary}
\begin{IEEEproof}
  See Appendix~\ref{sapp:spectral_converse}
\end{IEEEproof}

This coincides with maximum likelihood decoding over the fictitious channel
\begin{equation*}
  \Wsf = \trans{[\sigma_1(\abm_1), \hdots, \sigma_\Meu(\abm_\Meu)]}
\end{equation*}
under the input distribution $\pbm$, where $\sigma_i$ permutes the otherwise sorted spectrum in a way that the eigenvalues related to $\Ical_\Scal$ do not end up in the same column.

As an example, consider the depolarizing channel
\begin{equation*}
  \Ncal_p(\varphi) = (1-p) \varphi + \frac{p}{N} \Isf.
\end{equation*}
With some care in choosing the channel inputs (\cref{eq:tightness_depol_condition}), the spectral converse will be tight for $\Ncal_p$. This is formalized next.
\begin{corollary} \label{cor:depol_sc_tight}
  Let $\Scal$ be a set of channel inputs, if 
  \begin{equation} \label{eq:tightness_depol_condition}
    P_\ctn(\pbm, \Scal) = \sum_{(i,j) \in \Ical_{\Scal}} p_i a_{i, j},
  \end{equation}
  then $\Ical_{\Ncal_p \circ \Scal} = \Ical_\Scal$ and
  \begin{equation} \label{eq:tightness_depol}
    P_\ctn(\pbm, \Ncal_p \circ \Scal) = \sum_{(i,j) \in \Ical_\Scal} p_i \hat{a}_{i, j},
  \end{equation}
  where $\hat{\abm}_i = \tn{spec}(\Ncal_p(\varphi))$.
\end{corollary}
\begin{IEEEproof}
  See Appendix~\ref{sapp:depol_sc_tight}
\end{IEEEproof}

% === This helps Overleaf to understand how to collapse subsections === %

\subsection{Special Case: Classical Channel Coding Through Quantum Noise}

In this section, we require a set of symbols $\Scal = \brace{\rho_1, \dots, \rho_q} \subset \Dscr(\Hcal_q^{\otimes n})$ be a set of symbols with the sorted spectra $\abm_i = \tn{spec}(\rho_i)$, and spectral decomposition
$$\rho_i = \sum_{j \in [1:q]} a_{i, j} \ketbra{a_{i, j}}{a_{i, j}}.$$
Then, given a classical code $\Cscr = \brace{\xbm_m}_{m \in [1:\Meu]} \subseteq \Fbb_q^n$ with a priori probabilities $\pbm$, we are adapting \cref{thm:spectral_converse} to account for the structure of this coding scheme. First, note the following eigenbasis structure.

\begin{lemma}
  \label{lem:product_spectrum}
  Let $\ket{v} = \bigotimes_{\ell \in [1:n]} \ket{a_{i_\ell, j_\ell}}$ be a tensor product of eigenvectors, then
  \begin{equation*}
    \bigotimes_{\ell \in [1:n]} \varphi_{i_\ell} \cdot \ket{v} = \prod_{\ell \in [1:n]} a_{i_\ell, j_\ell} \cdot \ket{v}.
  \end{equation*}
  Furthermore, these pairs of eigenvectors and eigenvalues form the complete eigenbasis and spectrum.
\end{lemma}
\begin{IEEEproof}
  The first half is a direct consequence of the definition of the tensor product. The second holds by a cardinality argument.
\end{IEEEproof}

Then,  \cref{thm:spectral_converse} takes the shape
\begin{proposition} \label{prp:spectral_converse}
  Let $\Ical_{\Scal}$ be an index set of size $|\Ical_{\Scal}| = q^n$ and the property
  \begin{IEEEeqnarray}{rl}
    \forall (m, \jbm) \in \Ical_{\Scal}&, \forall (m', \jbm') \in \Ical_{\Scal}^c: \\
    &p_m \prod_{l \in [1:n]} a_{x_{m,l}, j_l} \geq p_{m'} \prod_{l \in [1:n]} a_{x_{m',l}, j_l'}.
  \end{IEEEeqnarray}
  Then
  \begin{equation*}
    P_{\ctn}(\pbm, \Scal_\Cscr) \leq \sum_{(m, \jbm) \in \Ical_{\Scal}} p_m \prod_{l \in [1:n]} a_{x_{m,l}, j_l}.
  \end{equation*}
\end{proposition}
\begin{IEEEproof}
  Direct application of \cref{thm:spectral_converse} with the spectrum from \cref{lem:product_spectrum}.
\end{IEEEproof}
\cref{prp:spectral_converse} is the sphere packing bound for any $(n, \Meu)_q$ code over the fictive channel
\begin{equation*}
  \Wsf = \trans{[\sigma_1(\abm_1), \hdots, \sigma_\Meu(\abm_\Meu)]}
\end{equation*}
and a priori distribution $\pbm$.

% === This helps Overleaf to understand how to collapse subsections === %

\subsection{Computational Complexities of Converse Bounds}
\label{sec:complexities}

We give an account on the computational complexity of the converse bound, as mentioned before our focus is on the fidelity bounds.

First, calculating the bounds in the general case, $\Meu$ messages chosen from $\Dscr(\Hcal_N)$. The spectral converse is quicker to calculate by a factor of $\Meu$ than the others, see \cref{tab:general}.
\begin{table}[t]
\begin{center}
  \caption{Computational complexities of converse bounds when there are no assumptions on the density operators used for CQ communication. The spectral converse gives a factor of $\Meu$ speed up. See the corresponding algorithms in Appendix~\ref{sapp:complexity_general}.} \label{tab:general}
  \vspace{-2ex}
  \begin{tabular}{cc}
    \toprule
    Converse bound & Complexity \\
    \midrule
    \cref{cor:Qiu_fidelity_bound}    & $\Ocal(\Meu^2 N^3)$ \\
    \cref{thm:Li_fidelity_bound}     & $\Ocal(\Meu^2 N^3)$ \\
    \cref{cor:Gemini_fidelity_bound} & $\Ocal(\Meu^2 N^3)$ \\
    \cref{thm:spectral_converse}     & $\Ocal(\Meu N^3)$ \\
    \bottomrule
  \end{tabular}
\end{center}
\end{table}
For the algorithms supporting the entries of \cref{tab:general} see Appendix~\ref{sapp:complexity_general}. These terms don't collapse evenly as we add symmetries to our communication model.

Now consider using a $(n,q^{r n})_q$ code with symbols $\Scal \subset \Dscr(\Hcal_q)$. Furthermore, we assume that the symbols are symmetric in the sense of \cref{def:symmetric_set}.
\begin{definition} \label{def:symmetric_set}
  $\Scal \subset \Dscr(\Hcal)$ is called symmetric if for all $\varphi_i, \varphi_j \in \Scal$ there is a unitary $\Usf_{ij}$ such that
  \begin{equation*}
    \varphi_i = \Usf_{ij} \varphi_j \hermi{\Usf_{ij}}.
  \end{equation*}
\end{definition}
This symmetry is sometimes used to simplify calculations (see, for example, \cite{BlascoCollVazquez-VilarFonollosa22_1}). Then the spectral converse becomes dependent only on the code parameters, and the complexities reduce to the ones in Table~\ref{tab:symmetric}. See Appendix~\ref{sapp:complexity_symmetric} for the algorithms behind the symmetric complexities.

\begin{table}[t]
\begin{center}
  \caption{Computational complexities of converse bounds with the symmetry defined in \cref{def:symmetric_set} assumed. Notably, the spectral converse is the single polynomial-time bound. For comparison, we added the General column, which contains the complexities, if one would use the algorithms that do not assume symmetry. See the corresponding algorithms in Appendix~\ref{sapp:complexity_symmetric}.} \label{tab:symmetric}
   \vspace{-2ex}
  \begin{tabular}{ccc}
    \toprule
    Converse bound & General & Symmetric \\
    \midrule
    \cref{cor:Qiu_fidelity_bound}    & $\Ocal(q^{(2r + 3) n})$ & $\Ocal(q^{2 r n} n)$ \\
    \cref{thm:Li_fidelity_bound}     & $\Ocal(q^{(2r + 3) n})$ & $\Ocal(q^{2 (r + 1) n} n)$ \\
    \cref{cor:Gemini_fidelity_bound} & $\Ocal(q^{(2r + 3) n})$ & $\Ocal(q^{2 r n} n)$ \\
    \cref{thm:spectral_converse}     & $\Ocal(q^{(r + 3) n})$ & $\Ocal(q^2 n)$ \\
    \bottomrule
  \end{tabular}
\end{center}
\end{table}

\begin{figure*}[t]
    \centering
    \subfloat[{$[3,1]_2$ repetition code $\Cscr_\textnormal{R}$.}]{%
        \begin{minipage}[t]{0.49\textwidth}
            \centering
            \resizebox{\linewidth}{!}{\input{figures/Repetition_n3_k1_q2/plot}}
        \end{minipage}%
        \label{fig:rep_n3}%
    } \hfill
    \subfloat[{$[6,1]_2$ repetition code $\Cscr_\textnormal{R}$.}]{%
        \begin{minipage}[t]{0.49\textwidth}
            \centering
            \resizebox{\linewidth}{!}{\input{figures/Repetition_n6_k1_q2/plot}}
        \end{minipage}%
        \label{fig:rep_n6}%
    } \\

    \subfloat[{$[3,2]_2$ SPC code $\Cscr_{\textnormal{SPC}}$.}]{%
        \begin{minipage}[t]{0.49\textwidth}
            \centering
            \resizebox{\linewidth}{!}{\input{figures/SPC_n3_k2_q2/plot}}
        \end{minipage}%
        \label{fig:spc_n3}%
        }\hfill
        \subfloat[{$[4,3]_2$ SPC code $\Cscr_{\textnormal{SPC}}$.}]{%
        \begin{minipage}[t]{0.49\textwidth}
            \centering
            \resizebox{\linewidth}{!}{\input{figures/SPC_n4_k3_q2/plot}}
        \end{minipage}%
        \label{fig:spc_n4}%
    }

    \subfloat[{$[5,4]_2$ SPC code $\Cscr_{\textnormal{SPC}}$.}]{%
        \begin{minipage}[t]{0.49\textwidth}
            \centering
            \resizebox{\linewidth}{!}{\input{figures/SPC_n5_k4_q2/plot}}
        \end{minipage}%
        \label{fig:spc_n5}%
        }\hfill
        \subfloat[{$[6,3]_2$ shortened Hamming code $\Cscr_\textnormal{H}$.}]{%
        \begin{minipage}[t]{0.49\textwidth}
            \centering
            \resizebox{\linewidth}{!}{\input{figures/Hamming_n6_k3_q2/plot}}
        \end{minipage}%
        \label{fig:ham_n6}%
    }
    \caption{Comparing converses for several short codes under amplitude damping noise}
    \label{fig:plots}
\end{figure*}
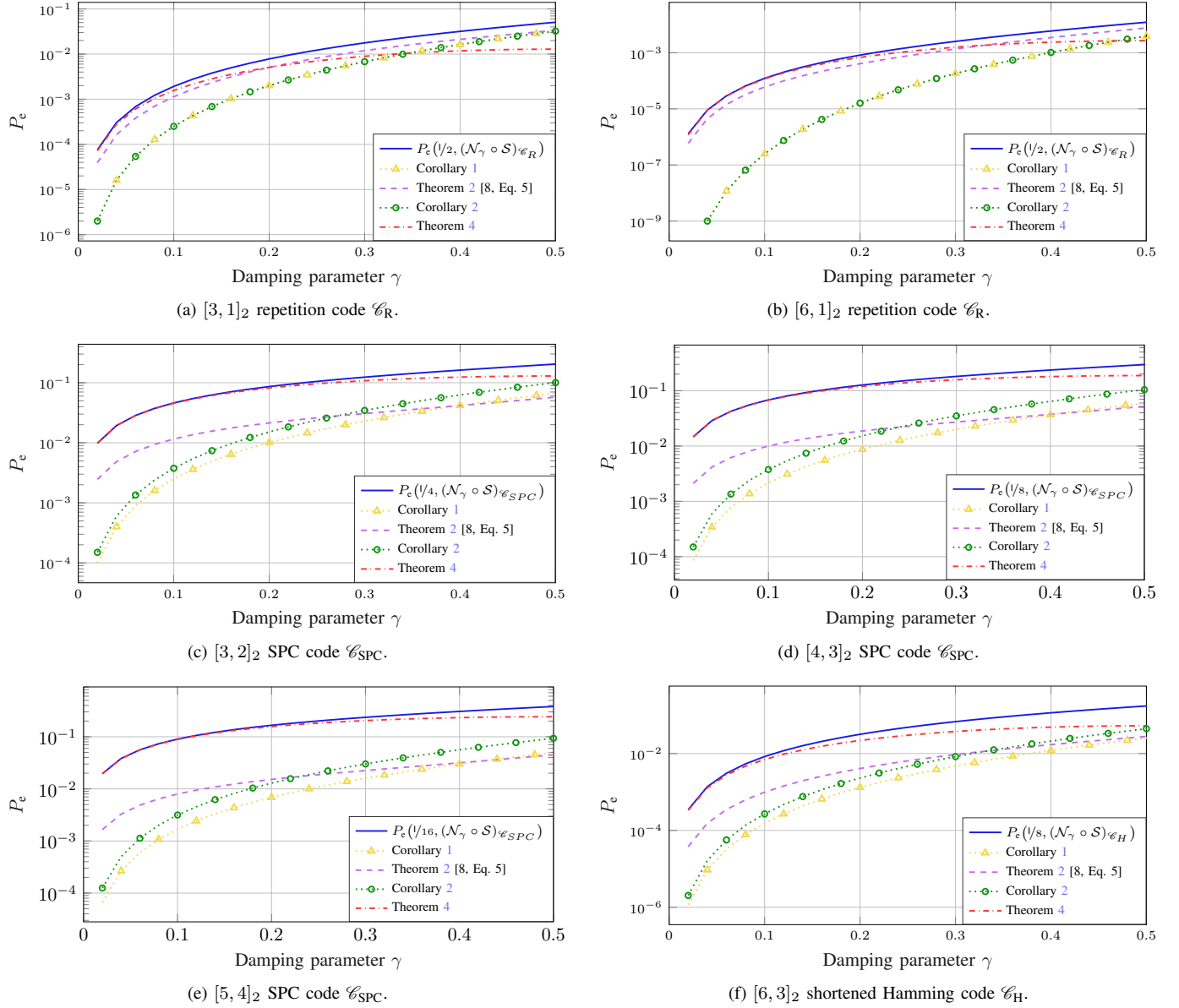

% === This helps Overleaf to understand how to collapse subsections === %

\subsection{Classical Channel Coding for Amplitude Damping Noise}
\label{sec:SC_channel_coding}

In the case of the ADC $\Ncal_\gamma$, \cref{lem:adc_specrtum} tells us that choosing $a = \nicefrac{1}{2}$ results in the symmetry of \cref{def:symmetric_set} regardless of the relative phase $f$. These are also the optimal channel inputs for the single-channel use case~\cite{HavasLinRosnesLai25_1}.
\begin{corollary}
  Let the channel input symbols be
  \begin{equation*}
    \Scal = \brace{\rho_0, \rho_1}= \brace{\frac{1}{2}
      \begin{bmatrix} 
        1 & e^{\ieu f} \\ e^{-\ieu f} & 1 
      \end{bmatrix}.
      \frac{1}{2}
      \begin{bmatrix}
        1 & -e^{\ieu f} \\ -e^{-\ieu f} & 1
      \end{bmatrix}}.
  \end{equation*}
  Then any $(n,\Meu)_2$ code $\Cscr$ over $\Ncal_\gamma$ abides by the converse
  \begin{IEEEeqnarray}{rCl}
    P_{\ctn}((\Ncal_\gamma \circ \Scal)_\Cscr)& \leq &\sum_{i \in [0:\ell]} \binom{n}{i} a_0^i a_1^{n-i}
    \IEEElinebreak{+} \paren{\frac{2^n}{\Meu} - L_\ell} a_0^{\ell+1} a_1^{n-\ell-1},
  \end{IEEEeqnarray}
  where $\tn{spec}(\rho_0) = \tn{spec}(\rho_1) = \trans{(a_0, a_1)}$, and
  \begin{equation*}
    L_\ell = \sum_{i \in [0:\ell]} \binom{n}{i} ,\quad L_\ell \leq \frac{2^n}{\Meu} < L_{\ell+1}.
  \end{equation*}
\end{corollary}
\begin{IEEEproof}
  Particular case of \cref{prp:spectral_converse}.
\end{IEEEproof}
  
In many cases~\cref{prp:spectral_converse} looks like a multinomial probability expression. These do not have closed forms in general, so very large code lengths $n$ would make computation difficult. However, since \cref{prp:spectral_converse} is the sphere packing bound for some fictive channel, one can use~\cite[Th.~49]{PolyanskiyPoorVerdu10_1} to calculate a normal approximation of the spectral converse for these code lengths $n$. On top of that, in this case, the fictitious channel is binary and symmetric, so the more accurate version of the statement yields an approximation that can be computed in constant time, since it does not contain summations or optimization anymore.
\begin{theorem}[{\cite[Th.~52]{PolyanskiyPoorVerdu10_1}}]
  Given a $(n, \Meu)_2$ code $\Cscr$, let
  \begin{equation*}
    \delta = \frac{1 - \sqrt{1 - \gamma (1 - \gamma)}}{2}
  \end{equation*}
  be the smaller eigenvalue. Then for any rate $r = \frac{\log_2(\Meu)}{n}$ above the capacity $C = 1 - \Hb(\delta)$ of the fictive channel, the success probability is
  \begin{IEEEeqnarray*}{rCl}
    P_\ctn\paren{\frac{\Isf}{\Meu}, (\Ncal_\gamma \circ \Scal)_\Cscr} = \Phi\paren{\frac{C - r - \frac{\log_2(n)}{2n}}{V}} + \Ocal\paren{\frac{e^{-\frac{n}{2}}}{n}},
  \end{IEEEeqnarray*}
  where $V = \sqrt{\frac{\delta(1-\delta)}{n}} \log_2 \frac{1-\delta}{\delta}$ and $\Phi$ is the standard normal cumulative distribution function.
\end{theorem}

\section{Numerical Results} \label{sec:numerical}
% \input{sections/Numerical}
% \subsection{Small Examples}

For low noise, the spectral converse gives a tighter bound than the other converses presented. On \cref{fig:plots}~\protect\subref{fig:rep_n3} and~\protect\subref{fig:rep_n6}, we plotted all curves, including the numerical solution to the semi-definite program for the $3$ and $6$-bit repetition codes. Similarly, \cref{fig:plots}~\protect\subref{fig:spc_n3},~\protect\subref{fig:spc_n4}, and \protect\subref{fig:spc_n5} show the bounds for the $3$, $4$, and $5$ bit-long single parity-check code (SPC). Lastly, \cref{fig:plots}~\protect\subref{fig:ham_n6} contains the curves for the shortened Hamming code defined by the generator matrix
\begin{equation*}
  \begin{bmatrix}
    1 & 0 & 0 & 1 & 1 & 0 \\
    0 & 1 & 0 & 1 & 0 & 1 \\
    0 & 0 & 1 & 0 & 1 & 1
  \end{bmatrix}.
\end{equation*}
In each case, the channel input symbols were $\Scal = \brace{\ket{\pm}}$ over the ADC $\Ncal_\gamma$.

\section{Conclusion and Further Directions} \label{sec:conclusion}
We presented converse bounds for multiple quantum hypothesis testing as a means of discussing probabilities in classical-quantum communication. The first is an improvement over an existing bound by Qiu, and the second, which we call the spectral converse, is a particular case of the underlying optimization problem. We proved tightness for the depolarizing channel under some natural conditions and demonstrated relative tightness for the ADC via numerical methods for small error-correcting codes when the channel has small noise. The spectral converse is easier to compute in terms of complexity than the other converse bounds for hypothesis testing, and it becomes significantly simpler as we introduce some symmetry into the channel model.
Because it can be seen as maximum-likelihood decoding over a fictitious classical channel, the spectral converse admits well-known normal approximations from the classical community.

It remains to be shown under what conditions each of these bounds is tighter than the others. Similarly, it is unknown how strong a restriction it is to consider only product states when sending classical information through quantum noise.

% \section*{Acknowledgment}

% We are indebted to Michael Shell for maintaining and improving
% \texttt{IEEEtran.cls}. 

%%%%%% 
%% To balance the columns at the last page of the paper use this
%% command:
%% 
% \enlargethispage{-1.2cm} 
%% 
%% If the balancing should occur in the middle of the references, use
%% the following trigger:
%% 

% \ifthenelse{\boolean{arXiv}}{\IEEEtriggeratref{4}}{
% \IEEEtriggeratref{15}
\balance % temporarily balancing the REFs  
% } 
%%   
%%   which triggers a \newpage (i.e., new column) just before the given
%%   reference number. Note that you need to adapt this if you modify
%%   the paper.  The "triggered" command can be changed if desired:
%%   
%   \IEEEtriggercmd{\enlargethispage{-20cm}}
%%   
%%%%%%   
%%%%%%   
%%   Appendix:
%%   If needed a single appendix is created by
%%   
%   \appendix
%%   
%%   If several appendices are needed, then the command
%%   
%   \appendices
%%   
%%   in combination with further \section commands can be used.
%%%%%%   
%%%%%%%%%%%%%%%%%%%%%%%%%%%%%%%%%%%%%%%%%%%%%%%%%%%%%%%%%%%%%%%%%%%%%   
\appendices

\section{Proof of Error Converse} \label{app:error_converse}
\subsection{Proof of \cref{thm:Gemini_trace_bound}} \label{sapp:Gemini_trace_bound}

Let $\Dcal = \brace{\Lambda_1, \dots, \Lambda_\Meu}$ be an arbitrary POVM for our discrimination problem, then
$$P_\ctn(\pbm, \Scal, \Dcal) = 1 - P_\etn(\pbm, \Scal, \Dcal) = \sum_{i \in [1:\Meu]} \trace{p_i \varphi_i \Lambda_i}.$$
Now, let $\Delta_{ij} = p_i \varphi_i - p_j \varphi_j$ and
$$S = \sum_{1 \leq i < j \leq \Meu} \trace{\Delta_{ij} (\Lambda_i - \Lambda_j)},$$
so we have
\begin{IEEEeqnarray*}{rCl}
  S& = &\sum_{i < j} p_i \trace{\varphi_i \Lambda_i} + \sum_{i < j} p_j \trace{\varphi_j \Lambda_j}
  \IEEElinebreak{-} \sum_{i < j} p_i \trace{\varphi_i \Lambda_j} - \sum_{i < j} p_j \trace{\varphi_j \Lambda_i} \\
  & = &(\Meu - 1) \sum_{i} p_i \trace{\varphi_i \Lambda_i} - \sum_{i \neq j} p_i \trace{\varphi_i \Lambda_j} \\
  & = &(\Meu - 1) P_\ctn(\pbm, \Scal, \Dcal) - P_\etn(\pbm, \Scal, \Dcal) \\
  & = &\Meu P_\ctn(\pbm, \Scal, \Dcal) - 1.
\end{IEEEeqnarray*}
On the other hand, we can bound the terms in $S$ by the decomposition $\Delta_{ij} = \Asf_{ij} - \Bsf_{ij}$, where $\Asf, \Bsf \geq 0$ and $\Asf \Bsf = 0$. Thus
\begin{IEEEeqnarray*}{rCl}
  \trace{\Delta_{ij} (\Lambda_i - \Lambda_j)} &=& \trace{\Asf_{ij} \Lambda_i} - \trace{\Asf_{ij} \Lambda_j}
  \IEEElinebreak{-} \trace{\Bsf_{ij} \Lambda_i} + \trace{\Bsf_{ij} \Lambda_j} \\
  & \leq &\trace{\Asf_{ij} \Lambda_i} + \trace{\Bsf_{ij} \Lambda_j} \\
  & \leq &\trace{\Asf_{ij}} + \trace{\Bsf_{ij}} = \tn{tr} \absol{\Delta_{ij}}.
\end{IEEEeqnarray*}
So $S = \Meu P_\ctn(\pbm, \Scal, \Dcal) - 1\leq \sum_{i < j} \tn{tr} \absol{\Delta_{ij}}$. Substituting $P_\etn = 1 - P_\ctn$ and division by $\Meu$ yields
\begin{IEEEeqnarray*}{rCl}
    \frac{1}{\Meu} \sum_{i < j} \tn{tr} \absol{\Delta_{ij}} &\geq& 1 - \frac{1}{\Meu} - P_\etn(\pbm, \Scal, \Dcal).
\end{IEEEeqnarray*}

\section{Proof of Spectral Converse and Related Statements}
\label{app:success_converse}
\subsection{Proof of \cref{thm:spectral_converse}}
\label{sapp:spectral_converse}

We can derive the spectral converse bound from \cref{lem:mqht_to_bqht}. Instead of the maximization in \cref{eq:bqht_max} we can set $\mu = \frac{\Isf}{2^n}$ to the maximally mixed, state and in place of taking the infimum in \cref{eq:bqht_type_I} we can choose
\begin{IEEEeqnarray*}{c}
  \Lambda = \diag{\Lambda_1, \dots, \Lambda_\Meu},
\end{IEEEeqnarray*}
where $\Lambda_i = \sum_{j: (i, j) \in \Ical_\Scal} \ketbra{a_{i, j}}{a_{i, j}}$. Thus, we have a lower bound on the error probability, which yields an upper bound on the success probability. The second inequality comes from
\begin{equation*}
  \sum_{(i,j) \in \Ical_\Scal} p_i a_{i,j} \leq \sum_{i \in [1:\Meu], j \in [1:2^n]} p_i a_{i,j} = 1.
\end{equation*}
Furthermore, this can be achieved whenever $\brace{\Lambda_i}_{i \in [1:\Meu]}$ forms a POVM.

% === This helps Overleaf to understand how to collapse subsections === %

\subsection{Proof of \cref{thm:spectral_converse_tight}}
\label{sapp:spectral_converse_tight}

The goal is to find an operator $t \mu$ such that the projectors $\brace{\varphi_i - t \mu}_+$ are orthogonal and
$$\sum_{m \in [1:\Meu]} \brace{\varphi_m - t\mu}_+ + \brace{\varphi_m - t\mu}_0 \geq \Isf,$$
given that $\Lambda_i = \sum_{j: (i,j) \in \Ical_\Scal} \ketbra{a_{i,j}}{a_{i,j}}$ forms a POVM.

First, note that $\Lambda_i \Lambda_{i'} = 0$ if $i \neq i'$, because the cardinality of $\Ical_\Scal$ is $N$ and
$$\sum_{(i,j) \in \Ical_\Scal} \ketbra{a_{i,j}}{a_{i,j}} = \Isf \in \Lscr(\Hcal_N).$$
Let $c \eqdef \min_{(i,j) \in \Ical_\Scal} p_i a_{i,j}$ the smallest eigenvalue and 
$$\varepsilon \eqdef \min(p_i a_{i,j} | (i,j) \in \Ical_\Scal, p_i a_{i,j} \neq c) - c$$
the gap to the next value. Furthermore, we define the projectors 
\begin{align*}
  \Bsf_i^\bullet &\eqdef \sum_{\substack{j: (i,j) \in \Ical_\Scal \\ a_{i,j} > c}} \paren{a_{i,j} - \frac{\varepsilon}{2}} \cdot \ketbra{a_{i,j}}{a_{i,j}} \\
            \Bsf_i^\circ &\eqdef \sum_{\substack{j: (i,j) \in \Ical_\Scal \\ a_{i,j} = c}} \ketbra{a_{i,j}}{a_{i,j}}
\end{align*}
and the operator
$$t \mu \eqdef \sum_{i \in [1:\Meu]} \Bsf_i^\bullet + \Bsf_i^\circ.$$
Then $\Bsf_i^\bullet, \Bsf_i^\circ \leq \Lambda_i$ and
\begin{align*}
  \brace{\varphi_i - t \mu}_+ &= \Lambda_i - \Bsf_i^\circ \\
      \brace{\varphi_i - t \mu}_0 &\geq \Bsf_i^\circ,
\end{align*}
which satisfy both conditions of quasi-perfect codes.

% === This helps Overleaf to understand how to collapse subsections === %

\subsection{Proof of \cref{cor:depol_sc_tight}} \label{sapp:depol_sc_tight}

We will show that the depolarizing channel has little impact on both the spectrum and the eigenspaces.

First of all, note if $\varphi = \sum_{a \in \tn{spec}(\varphi)} a \cdot \ketbra{a}{a}$ is the spectral decomposition of a Hermitian operator $\varphi \in \Lscr(\Hcal)$ then for any $c \in \Rbb$
$$\varphi - c \cdot \Isf = \sum_{a \in \tn{spec}(\varphi)} (a - c) \cdot \ketbra{a}{a}.$$
Thus, the depolarizing channel does not change the order of the spectrum and does not change the eigenvectors, i.e., $\Ical_\Scal = \Ical_{\Ncal_p \circ \Scal}$.

Second, \cref{eq:tightness_depol_condition} is equivalent to the tightness of the spectral converse, the following condition must hold
$$\sum_{j: (i,j) \in \Ical_\Scal} \ketbra{a_{i,j}}{a_{i,j}} \cdot \sum_{j': (i',j') \in \Ical_\Scal} \ketbra{a_{i',j'}}{a_{i',j'}} = 0,$$
where $i \neq i'$. Since $\Ncal_p$ keeps the eigenspaces in place, this condition is satisfied by the channel outputs as well. Therefore, the spectral converse is tight for the outputs, and \cref{eq:tightness_depol} holds.

\section{Algorithms for Calculating Converses}
\label{app:complexities}
In every case, we will need to compute spectra, sometimes using eigenvectors for Hermitian matrices. Both with and without the eigenvectors, the computation takes $\Ocal(N^3)$ operations for a matrix of size $N \times N$. See \cite[Ch.~5]{Demmel97_1} for details. Furthermore, calculating the fidelity between two Hermitian operators requires at least one spectrum, thus has the same complexity $\Ocal(N^3)$.

\subsection{General Case\textemdash Without Assumptions} \label{sapp:complexity_general}

For \cref{cor:Qiu_fidelity_bound} and \cref{cor:Gemini_fidelity_bound}, the pairwise fidelities must be calculated separately, because there is no structure in the states that can be exploited. This means we have to make $\Ocal(N^3)$ operations $\Ocal(\Meu^2)$ times.

When calculating \cref{thm:Li_fidelity_bound}, one can precompute all eigen decompositions for $\Ocal(\Meu N^3)$ operations. However, there are four nested loops in the computation for a total of $\Ocal(\Meu^2 N^2)$ iterations, and in each of them, we must calculate an inner-product of length $N$ that has not been calculated before.

Last, \cref{thm:spectral_converse} requires the precomputation of the spectra for $\Ocal(\Meu N^3)$ operations. Then we have to find the $N$ largest values out of $\Meu N$ numbers, which can be done in $\Ocal(\Meu N \log_2 N)$ steps, which means that the precomputation dominates the execution.

% === This helps Overleaf to understand how to collapse subsections === %

\subsection{Symmetric Case\textemdash Assuming \cref{def:symmetric_set}} \label{sapp:complexity_symmetric}

Once again, \cref{cor:Qiu_fidelity_bound} and \cref{cor:Gemini_fidelity_bound} are mostly the same. There is a \hlgreen{pair-wise} precomputation of \hlblue{symbol fidelities} in $\Scal$ for $\Ocal(\hlgreen{q^2} \cdot \hlblue{q^3})$ steps, then for pairs of messages ($\Ocal(\Meu^2) = \Ocal(q^{2rn})$) we must compute the message fidelities ($\Ocal(n)$) based on the precomputed data. Here, we used the fact that the fidelity commutes with tensor products. This time the main iteration dominates with a complexity of $\Ocal(q^{2rn} n)$ if
$$2rn\log_q(n) \geq 5.$$

Similarly, for \cref{thm:Li_fidelity_bound}, we can exploit that the inner product commutes with the tensor product by definition. Thus, we can compute the spectral decompositions first ($\Ocal(q^4)$), then find all symbol inner products ($\Ocal(q^5)$). In the main loop, we must go through the pairs of messages ($\Ocal(q^{2rn})$) and each pair of eigenvectors ($\Ocal(q^{2n})$) to then calculate the minimum of eigenvalues and the message inner product ($\Ocal(n)$). The main loop dominates for
$$2 (r + 1) n \log_q(n) \geq 5.$$

Finally, \cref{thm:spectral_converse} requires a single spectrum to begin ($\Ocal(q^3)$), which should be sorted ($\Ocal(q \log_2(q))$) to get $\abm$. Then we want to calculate the vector
$$\vbm = \trans{\paren{\frac{a_2}{a_1}, \frac{a_3}{a_2}, \dots, \frac{a_q}{a_{q-1}}}}$$
and initialize the current 'error pattern' $\lbm= \trans{(n, 0, \cdots, 0)} \in \Nbb^q$, corresponding 'probability' $p = a_1^n$, cardinality $c = \binom{n}{n} = 1$, left over space to fill $L = q^{(1 - r) n} - c$ and accumulated probability $P = cp$. Then, for each iteration, we calculate
$$i = \argmax_{j \in [1:q-1]} \paren{v_j | l_j > 0}$$
and update the variables
\begin{align*}
  l_i &= l_i - 1 \\
  l_{i+1} &= l_{i+1} + 1 \\
  p &= p \cdot v_i \\
  c &= c \cdot \frac{l_i - 1}{l_{i+1}} \\
  L &= L - c \\
  P &= P + cp.
\end{align*}
Until we find $L < 0$, and we correct $P = P + Lp$. Each iteration may take up to $\Ocal(q)$ step,s and the breaking condition $L < 0$ must occur before we reach
$$\lbm= \trans{(0, \dots, 0, n)},$$
That is much less than $qn$ iterations.

% \ifthenelse{\boolean{JOURNAL}}{\input{}}{}

%%%%%% 
%% References:
%% We recommend the usage of BibTeX:
%% 

\bibliographystyle{IEEEtran}
\bibliography{defshort1, biblioHY}

\end{document}

%% file: figures/Repetition_n3_k1_q2/plot.tex
    \begin{tikzpicture}
        \begin{axis}[ 
                width = \linewidth, height = .5\linewidth,
                scale only axis,
                %title = {$3$ bit repetition code over $\Ncal_\gamma$ with input symbols $\ket{\pm}$},
                % --- X axis --- %
                xmin = 0, xmax = 0.5,  % xmax = 1,
                xlabel style = {font = \color{white!10!black}},
                xlabel = {Damping parameter $\gamma$},
                xmajorgrids, xtick distance = .1,
                % --- Y axis --- %
                %ymin = 0.00001, ymax = 1.02,
                yminorticks = true, 
                ymode = log,
                ylabel style = {font = \color{white!10!black}}, % yshift = -6mm can be added to ylabel style
                ylabel = {$P_{\tn{e}}$},
                ymajorgrids, % ytick distance = .2,
                % --- Both axes --- %
                tick label style={font=\scriptsize},
                % --- Legend --- %
                legend style = {draw = none, fill = white, 
                    font = \scriptsize, at = {(1, 0)}, anchor = south east, legend cell align = left, align = left, draw = white!20!black}
            ]

            % --- Error Probability --- %
            \addplot [color = BrightBlue, line width = .3mm] 
                table [ x index = 0,
                        y expr = \thisrowno{1} % <-- compute on the fly
                      ] {figures/Repetition_n3_k1_q2/P_err_optimal.txt};
            \addlegendentry{$P_{\textnormal{e}}\paren{\nicefrac{\Isf}{2}, (\Ncal_\gamma \circ \Scal)_{\Cscr_R}}$}

            % --- Converses --- %
            \addplot [dotted, color = BrightYellow, line width = .3mm, mark size=2.0pt, mark=triangle, mark options={solid, BrightYellow}, mark repeat=2, mark phase=2] 
                table [ x index = 0,
                        y expr = \thisrowno{1} % <-- compute on the fly
                      ] {figures/Repetition_n3_k1_q2/Conv_Qiu.txt};
            \addlegendentry{\cref{cor:Qiu_fidelity_bound}}
            
            \addplot [dashed, color = LightPurple, line width = .3mm] 
                table [ x index = 0,
                        y expr = \thisrowno{1} % <-- compute on the fly
                      ] {figures/Repetition_n3_k1_q2/Conv_Li.txt};
            \addlegendentry{\cref{thm:Li_fidelity_bound} \cite[Eq.~5]{Li16_1}}
            
            \addplot [dotted, color = DarkGreen, line width = .3mm, mark size=1.5pt, mark=o, mark options={solid, DarkGreen}, mark repeat=2] 
                table [ x index = 0,
                        y expr = \thisrowno{1} % <-- compute on the fly
                      ] {figures/Repetition_n3_k1_q2/Conv_Gemini.txt};
            \addlegendentry{\cref{cor:Gemini_fidelity_bound}}
            
            \addplot [dashdotted, color = BrightRed, line width = .3mm] 
                table [ x index = 0,
                        y expr = \thisrowno{1} % <-- compute on the fly
                      ] {figures/Repetition_n3_k1_q2/Conv_spectral.txt};
            \addlegendentry{\cref{thm:spectral_converse}}
            
        \end{axis}
    \end{tikzpicture}

%% file: figures/Repetition_n6_k1_q2/plot.tex
    \begin{tikzpicture}
        \begin{axis}[ 
                width = \linewidth, height = .5\linewidth,
                scale only axis,
                %title = {$6$ bit repetition code over $\Ncal_\gamma$ with input symbols $\ket{\pm}$},
                % --- X axis --- %
                xmin = 0, xmax = 0.5,  % xmax = 1,
                xlabel style = {font = \color{white!10!black}},
                xlabel = {Damping parameter $\gamma$},
                xmajorgrids, xtick distance = .1,
                % --- Y axis --- %
                %ymin = 0.000001, ymax = 1.02,
                yminorticks = true, 
                ymode = log,
                ylabel style = {font = \color{white!10!black}}, % yshift = -6mm can be added to ylabel style
                ylabel = {$P_{\tn{e}}$},
                ymajorgrids, % ytick distance = .2,
                                % --- Both axes --- %
                tick label style={font=\scriptsize},
                % --- Legend --- %
                legend style = {draw = none, fill = white, 
                    font = \scriptsize, at = {(1, 0)}, anchor = south east, legend cell align = left, align = left, draw = white!20!black}
            ]

            % --- Error Probability --- %
            \addplot [color = BrightBlue, line width = .3mm] 
                table [ x index = 0,
                        y expr = \thisrowno{1} % <-- compute on the fly
                      ] {figures/Repetition_n6_k1_q2/P_err_optimal.txt};
            \addlegendentry{$P_{\textnormal{e}}\paren{\nicefrac{\Isf}{2}, (\Ncal_\gamma \circ \Scal)_{\Cscr_R}}$}

            % --- Converses --- %
            \addplot [dotted, color = BrightYellow, line width = .3mm, mark size=2.0pt, mark=triangle, mark options={solid, BrightYellow}, mark repeat=2, mark phase=2] 
                table [ x index = 0,
                        y expr = \thisrowno{1} % <-- compute on the fly
                      ] {figures/Repetition_n6_k1_q2/Conv_Qiu.txt};
            \addlegendentry{\cref{cor:Qiu_fidelity_bound}}
            
            \addplot [dashed, color = LightPurple, line width = .3mm] 
                table [ x index = 0,
                        y expr = \thisrowno{1} % <-- compute on the fly
                      ] {figures/Repetition_n6_k1_q2/Conv_Li.txt};
            \addlegendentry{\cref{thm:Li_fidelity_bound} \cite[Eq.~5]{Li16_1}}
            
            \addplot [dotted, color = DarkGreen, line width = .3mm, mark size=1.5pt, mark=o, mark options={solid, DarkGreen}, mark repeat=2] 
                table [ x index = 0,
                        y expr = \thisrowno{1} % <-- compute on the fly
                      ] {figures/Repetition_n6_k1_q2/Conv_Gemini.txt};
            \addlegendentry{\cref{cor:Gemini_fidelity_bound}}
            
            \addplot [dashdotted, color = BrightRed, line width = .3mm] 
                table [ x index = 0,
                        y expr = \thisrowno{1} % <-- compute on the fly
                      ] {figures/Repetition_n6_k1_q2/Conv_spectral.txt};
            \addlegendentry{\cref{thm:spectral_converse}}
            
        \end{axis}
    \end{tikzpicture}

%% file: figures/SPC_n3_k2_q2/plot.tex
    \begin{tikzpicture}
        \begin{axis}[ 
                width = \linewidth, height = .5\linewidth,
                scale only axis,
                %title = {$3$ bit SPC code over $\Ncal_\gamma$ with input symbols $\ket{\pm}$},
                % --- X axis --- %
                xmin = 0, xmax = 0.5,  % xmax = 1,
                xlabel style = {font = \color{white!10!black}},
                xlabel = {Damping parameter $\gamma$},
                xmajorgrids, xtick distance = .1,
                % --- Y axis --- %
                %ymin = 0.0001, ymax = 1.02,
                yminorticks = true, 
                ymode = log,
                ylabel style = {font = \color{white!10!black}}, % yshift = -6mm can be added to ylabel style
                ylabel = {$P_{\tn{e}}$},
                ymajorgrids, % ytick distance = .2,
                % --- Both axes --- %
                tick label style={font=\scriptsize},
                % --- Legend --- %
                legend style = {draw = none, fill = white, 
                    font = \scriptsize, at = {(1, 0)}, anchor = south east, legend cell align = left, align = left, draw = white!20!black}
            ]

            % --- Error Probability --- %
            \addplot [color = BrightBlue, line width = .3mm] 
                table [ x index = 0,
                        y expr = \thisrowno{1} % <-- compute on the fly
                      ] {figures/SPC_n3_k2_q2/P_err_optimal.txt};
            \addlegendentry{$P_{\textnormal{e}}\paren{\nicefrac{\Isf}{4}, (\Ncal_\gamma \circ \Scal)_{\Cscr_{SPC}}}$}

            % --- Converses --- %
            \addplot [dotted, color = BrightYellow, line width = .3mm, mark size=2.0pt, mark=triangle, mark options={solid, BrightYellow}, mark repeat=2, mark phase=2] 
                table [ x index = 0,
                        y expr = \thisrowno{1} % <-- compute on the fly
                      ] {figures/SPC_n3_k2_q2/Conv_Qiu.txt};
            \addlegendentry{\cref{cor:Qiu_fidelity_bound}}
            
            \addplot [dashed, color = LightPurple, line width = .3mm] 
                table [ x index = 0,
                        y expr = \thisrowno{1} % <-- compute on the fly
                      ] {figures/SPC_n3_k2_q2/Conv_Li.txt};
            \addlegendentry{\cref{thm:Li_fidelity_bound} \cite[Eq.~5]{Li16_1}}
            
            \addplot [dotted, color = DarkGreen, line width = .3mm, mark size=1.5pt, mark=o, mark options={solid, DarkGreen}, mark repeat=2] 
                table [ x index = 0,
                        y expr = \thisrowno{1} % <-- compute on the fly
                      ] {figures/SPC_n3_k2_q2/Conv_Gemini.txt};
            \addlegendentry{\cref{cor:Gemini_fidelity_bound}}
            
            \addplot [dashdotted, color = BrightRed, line width = .3mm] 
                table [ x index = 0,
                        y expr = \thisrowno{1} % <-- compute on the fly
                      ] {figures/SPC_n3_k2_q2/Conv_spectral.txt};
            \addlegendentry{\cref{thm:spectral_converse}}
            
        \end{axis}
    \end{tikzpicture}

%% file: figures/SPC_n4_k3_q2/plot.tex
    \begin{tikzpicture}
        \begin{axis}[ 
                width = \linewidth, height = .5\linewidth,
                scale only axis,
                %title = {$4$ bit single SPC over $\Ncal_\gamma$ with input symbols $\ket{\pm}$},
                % --- X axis --- %
                xmin = 0, xmax = 0.5,  % xmax = 1,
                xlabel style = {font = \color{white!10!black}},
                xlabel = {Damping parameter $\gamma$},
                xmajorgrids, xtick distance = .1,
                % --- Y axis --- %
                %ymin = 0.0001, ymax = 1.02,
                yminorticks = true, 
                ymode = log,
                ylabel style = {font = \color{white!10!black}}, % yshift = -6mm can be added to ylabel style
                ylabel = {$P_{\tn{e}}$},
                ymajorgrids, % ytick distance = .2,
                % --- Legend --- %
                legend style = {draw = none, fill = white, 
                    font = \scriptsize, at = {(1, 0)}, anchor = south east, legend cell align = left, align = left, draw = white!20!black}
            ]

            % --- Error Probability --- %
            \addplot [color = BrightBlue, line width = .3mm] 
                table [ x index = 0,
                        y expr = \thisrowno{1} % <-- compute on the fly
                      ] {figures/SPC_n4_k3_q2/P_err_optimal.txt};
            \addlegendentry{$P_{\textnormal{e}}\paren{\nicefrac{\Isf}{8}, (\Ncal_\gamma \circ \Scal)_{\Cscr_{SPC}}}$}

            % --- Converses --- %
            \addplot [dotted, color = BrightYellow, line width = .3mm, mark size=2.0pt, mark=triangle, mark options={solid, BrightYellow}, mark repeat=2, mark phase=2] 
                table [ x index = 0,
                        y expr = \thisrowno{1} % <-- compute on the fly
                      ] {figures/SPC_n4_k3_q2/Conv_Qiu.txt};
            \addlegendentry{\cref{cor:Qiu_fidelity_bound}}
            
            \addplot [dashed, color = LightPurple, line width = .3mm] 
                table [ x index = 0,
                        y expr = \thisrowno{1} % <-- compute on the fly
                      ] {figures/SPC_n4_k3_q2/Conv_Li.txt};
            \addlegendentry{\cref{thm:Li_fidelity_bound} \cite[Eq.~5]{Li16_1}}
            
            \addplot [dotted, color = DarkGreen, line width = .3mm, mark size=1.5pt, mark=o, mark options={solid, DarkGreen}, mark repeat=2] 
                table [ x index = 0,
                        y expr = \thisrowno{1} % <-- compute on the fly
                      ] {figures/SPC_n4_k3_q2/Conv_Gemini.txt};
            \addlegendentry{\cref{cor:Gemini_fidelity_bound}}
            
            \addplot [dashdotted, color = BrightRed, line width = .3mm] 
                table [ x index = 0,
                        y expr = \thisrowno{1} % <-- compute on the fly
                      ] {figures/SPC_n4_k3_q2/Conv_spectral.txt};
            \addlegendentry{\cref{thm:spectral_converse}}
            
        \end{axis}
    \end{tikzpicture}

%% file: figures/SPC_n5_k4_q2/plot.tex
    \begin{tikzpicture}
        \begin{axis}[ 
                width = \linewidth, height = .5\linewidth,
                scale only axis,
                %title = {$5$ bit SPC code over $\Ncal_\gamma$ with input symbols $\ket{\pm}$},
                % --- X axis --- %
                xmin = 0, xmax = 0.5,  % xmax = 1,
                xlabel style = {font = \color{white!10!black}},
                xlabel = {Damping parameter $\gamma$},
                xmajorgrids, xtick distance = .1,
                % --- Y axis --- %
                %ymin = 0.0001, ymax = 1.02,
                yminorticks = true, 
                ymode = log,
                ylabel style = {font = \color{white!10!black}}, % yshift = -6mm can be added to ylabel style
                ylabel = {$P_{\tn{e}}$},
                ymajorgrids, % ytick distance = .2,
                % --- Legend --- %
                legend style = {draw = none, fill = white, 
                    font = \scriptsize, at = {(1, 0)}, anchor = south east, legend cell align = left, align = left, draw = white!20!black}
            ]

            % --- Error Probability --- %
            \addplot [color = BrightBlue, line width = .3mm] 
                table [ x index = 0,
                        y expr = \thisrowno{1} % <-- compute on the fly
                      ] {figures/SPC_n5_k4_q2/P_err_optimal.txt};
            \addlegendentry{$P_{\textnormal{e}}\paren{\nicefrac{\Isf}{16}, (\Ncal_\gamma \circ \Scal)_{\Cscr_{SPC}}}$}

            % --- Converses --- %
            \addplot [dotted, color = BrightYellow, line width = .3mm, mark size=2.0pt, mark=triangle, mark options={solid, BrightYellow}, mark repeat=2, mark phase=2] 
                table [ x index = 0,
                        y expr = \thisrowno{1} % <-- compute on the fly
                      ] {figures/SPC_n5_k4_q2/Conv_Qiu.txt};
            \addlegendentry{\cref{cor:Qiu_fidelity_bound}}
            
            \addplot [dashed, color = LightPurple, line width = .3mm] 
                table [ x index = 0,
                        y expr = \thisrowno{1} % <-- compute on the fly
                      ] {figures/SPC_n5_k4_q2/Conv_Li.txt};
            \addlegendentry{\cref{thm:Li_fidelity_bound} \cite[Eq.~5]{Li16_1}}
            
            \addplot [dotted, color = DarkGreen, line width = .3mm, mark size=1.5pt, mark=o, mark options={solid, DarkGreen}, mark repeat=2] 
                table [ x index = 0,
                        y expr = \thisrowno{1} % <-- compute on the fly
                      ] {figures/SPC_n5_k4_q2/Conv_Gemini.txt};
            \addlegendentry{\cref{cor:Gemini_fidelity_bound}}
            
            \addplot [dashdotted, color = BrightRed, line width = .3mm] 
                table [ x index = 0,
                        y expr = \thisrowno{1} % <-- compute on the fly
                      ] {figures/SPC_n5_k4_q2/Conv_spectral.txt};
            \addlegendentry{\cref{thm:spectral_converse}}
            
        \end{axis}
    \end{tikzpicture}

%% file: figures/Hamming_n6_k3_q2/plot.tex
    \begin{tikzpicture}
        \begin{axis}[ 
                width = \linewidth, height = .5\linewidth,
                scale only axis,
                %title = {Punctured Hamming code over $\Ncal_\gamma$ with input symbols $\ket{\pm}$},
                % --- X axis --- %
                xmin = 0, xmax = 0.5,  % xmax = 1,
                xlabel style = {font = \color{white!10!black}},
                xlabel = {Damping parameter $\gamma$},
                xmajorgrids, xtick distance = .1,
                % --- Y axis --- %
                %ymin = 0.00001, ymax = 1.02,
                yminorticks = true, 
                ymode = log,
                ylabel style = {font = \color{white!10!black}}, % yshift = -6mm can be added to ylabel style
                ylabel = {$P_{\tn{e}}$},
                ymajorgrids, % ytick distance = .2,
                % --- Both axes --- %
                tick label style={font=\scriptsize},
                % --- Legend --- %
                legend style = {draw = none, fill = white, 
                    font = \scriptsize, at = {(1, 0)}, anchor = south east, legend cell align = left, align = left, draw = white!20!black}
            ]

            % --- Error Probability --- %
            \addplot [color = BrightBlue, line width = .3mm] 
                table [ x index = 0,
                        y expr = \thisrowno{1} % <-- compute on the fly
                      ] {figures/Hamming_n6_k3_q2/P_err_optimal.txt};
            \addlegendentry{$P_{\textnormal{e}}\paren{\nicefrac{\Isf}{8}, (\Ncal_\gamma \circ \Scal)_{\Cscr_H}}$}

            % --- Converses --- %
            \addplot [dotted, color = BrightYellow, line width = .3mm, mark size=2.0pt, mark=triangle, mark options={solid, BrightYellow}, mark repeat=2, mark phase=2] 
                table [ x index = 0,
                        y expr = \thisrowno{1} % <-- compute on the fly
                      ] {figures/Hamming_n6_k3_q2/Conv_Qiu.txt};
            \addlegendentry{\cref{cor:Qiu_fidelity_bound}}
            
            \addplot [dashed, color = LightPurple, line width = .3mm] 
                table [ x index = 0,
                        y expr = \thisrowno{1} % <-- compute on the fly
                      ] {figures/Hamming_n6_k3_q2/Conv_Li.txt};
            \addlegendentry{\cref{thm:Li_fidelity_bound} \cite[Eq.~5]{Li16_1}}
            
            \addplot [dotted, color = DarkGreen, line width = .3mm, mark size=1.5pt, mark=o, mark options={solid, DarkGreen}, mark repeat=2] 
                table [ x index = 0,
                        y expr = \thisrowno{1} % <-- compute on the fly
                      ] {figures/Hamming_n6_k3_q2/Conv_Gemini.txt};
            \addlegendentry{\cref{cor:Gemini_fidelity_bound}}
            
            \addplot [dashdotted, color = BrightRed, line width = .3mm] 
                table [ x index = 0,
                        y expr = \thisrowno{1} % <-- compute on the fly
                      ] {figures/Hamming_n6_k3_q2/Conv_spectral.txt};
            \addlegendentry{\cref{thm:spectral_converse}}
            
        \end{axis}
    \end{tikzpicture}